\documentclass[aps,pra,twocolumn,superscriptaddress,groupedaddress,nofootinbib]{revtex4}  % for review and submission
\usepackage{graphicx}  % needed for figures
\usepackage{dcolumn}   % needed for some tables
\usepackage{bm}        % for math
\usepackage{amssymb}   % for math

\usepackage{hhline}
\usepackage{scalefnt}
\usepackage{amsmath}
\usepackage{amsthm}
\usepackage{color}
\usepackage{xr}
\usepackage{cases}
\usepackage{enumerate}

\usepackage{tikz}
\usetikzlibrary{matrix,arrows}

\usepackage{slashbox}
 \usepackage{enumitem}
 \usepackage{relsize}
\usepackage{xy}
\xyoption{matrix}
\xyoption{frame}
\xyoption{arrow}
\xyoption{arc}

\usepackage{ifpdf}
\ifpdf
\else
\PackageWarningNoLine{Qcircuit}{Qcircuit is loading in Postscript mode.  The Xy-pic options ps and dvips will be loaded.  If you wish to use other Postscript drivers for Xy-pic, you must modify the code in Qcircuit.tex}
\xyoption{ps}
\xyoption{dvips}
\fi

\entrymodifiers={!C\entrybox}

\newcommand{\conjTrans}{{\ast}}
\newcommand{\M}{s}
\newcommand{\N}{t}
\newcommand{\K}{r}
\newcommand{\MM}{d_A}
\newcommand{\NN}{d_B}
\newcommand{\LL}{d_{BC}}
\newcommand{\LLi}{d_{BC_i}}
\newcommand{\CC}{d_C}
\newcommand{\CCi}{d_{C_i}}
\newcommand{\p}{p}
\newcommand{\ket}[1]{\left| #1 \right>} % for Dirac bras
\newcommand{\bra}[1]{\left< #1 \right|} % for Dirac kets

\newcommand{\ketbra}[2]{|#1\rangle\!\langle#2|}
\newcommand{\proj}[1]{|#1\rangle\!\langle#1|}
\newcommand{\id}{\leavevmode\hbox{\small1\normalsize\kern-.33em1}}
\newcommand{\tr}{\mathrm{tr}}

\newcommand{\ot}{\otimes}
\newcommand{\cI}{\mathcal{I}}

\newcommand\restr[2]{{% we make the whole thing an ordinary symbol
  \left.\kern-\nulldelimiterspace % automatically resize the bar with \right
  #1 % the function
  \vphantom{\big|} % pretend it's a little taller at normal size
  \right|_{#2} % this is the delimiter
  }}
  
  \newcommand{\RM}[1]{\MakeUppercase{\romannumeral #1{}}}

\newtheorem{thm}{Theorem}
\newtheorem{prop}[thm]{Proposition}
\newtheorem{lem}[thm]{Lemma}

\newenvironment{proof1}[1][Proof]{\noindent\textbf{#1.} }{}

\theoremstyle{definition}
\newtheorem{defi}{Definition}
\newtheorem{rmk}{Remark}

\begin{document}

\title{Smooth Manifold Structure for Extreme Channels}
\author{Raban~Iten} 
\email{itenr@ethz.ch} 
\affiliation{Institute for Theoretical Physics, ETH Zurich, Switzerland}  
\author{Roger~Colbeck} 
\email{roger.colbeck@york.ac.uk} 
\affiliation{Department of Mathematics, University of York, YO10 5DD, UK}

\date{$25^{\textrm{th}}$ September 2019}

\begin{abstract}
  A quantum channel from a system $A$ of dimension $\MM$ to a system
  $B$ of dimension $\NN$ is a completely positive trace-preserving map
  from complex $\MM \times \MM$ to $\NN \times \NN$ matrices, and the
  set of all such maps with Kraus rank $\K$ has the structure of a
  smooth manifold. We describe this set in two ways. First, as a
  quotient space of (a subset of) the $\K \NN\times \MM$ dimensional
  Stiefel manifold. Secondly, as the set of all Choi-states of a fixed
  rank $\K$. These two descriptions are topologically equivalent. This
  allows us to show that the set of all Choi-states corresponding to extreme channels from system
  $A$ to system $B$ of a fixed Kraus rank $\K$ is a smooth submanifold
  of dimension $2\K\MM\NN-\MM^2-\K^2$ of the set of all Choi-states of
  rank $\K$. As an application, we derive a lower bound on the number
  of parameters required for a quantum circuit topology to be able to
  approximate all extreme channels from $A$ to $B$ arbitrarily well.
\end{abstract}

\maketitle

\section{Introduction}
We describe the differential structure of the set
$\mathcal{E}_{\M,\N,\K}$ consisting of all completely positive
trace-preserving (CPTP) maps from $\mathbb{C}^{\M \times \M}$ to
$\mathbb{C}^{\N \times \N}$ (which we refer to as $\M$ to $\N$
channels) of fixed Kraus rank $\K$. A linear map $\mathcal{E}:
\mathbb{C}^{\M \times \M}\mapsto \mathbb{C}^{\N \times \N}$ is called
positive if it sends positive semi-definite matrices to positive
semi-definite matrices. It is called completely positive (CP) if
$\mathcal{E}\otimes \cI_{\p}$ is positive for all $\p \in \mathbb{N}$,
where $\cI_{\p}:\mathbb{C}^{\p\times\p}\mapsto
\mathbb{C}^{\p\times\p}$ denotes the identity
channel. Choi~\cite{choi} showed that a map $\mathcal{E}$ is
completely positive if and only if it admits an expression
$\mathcal{E}(X)=\sum_{i=1}^{\K'} A_i X A_i^{\conjTrans}$ (for all $X
\in \mathbb{C}^{\M \times \M}$), where the $A_i \in \mathbb{C}^{\N
  \times \M}$ are called Kraus operators in quantum information
theory~\cite{Kraus}. The Kraus representation is not unique in general
and the minimum number of Kraus operators $\K'$, such that a
representation of the form above exists, is called the Kraus rank $\K$
of the map $\mathcal{E}$ (and the corresponding representation is
called a `minimal' Kraus representation).  Note that, by Remark 4
of~\cite{choi}, a Kraus representation is minimal
if and only if the Kraus operators $A_1, A_2,\dots ,A_{\K'}$ are linearly independent. Finally, a map $\mathcal{E}$ is called trace preserving if $\tr \, \mathcal{E}(X)=\tr \,X$ for all $X \in \mathbb{C}^{\M \times \M}$, which corresponds to the requirement $\sum_{i=1}^{\K} A_i^{\conjTrans}A_i=I$ on the Kraus operators. \\

CPTP maps are of interest in physics, because they describe the most
general evolution  a quantum system can undergo. Since the set
$\mathcal{E}_{\M,\N}$ of all $\M$ to $\N$ quantum channels   is
convex, one can investigate the decomposition of a quantum channel
into a convex combination of extreme channels. In particular, such
decompositions can help to implement quantum channels in a cheaper
way~\cite{Channels, Wang_qubit}. However, there are open questions
about the structure of the (closure of the) set of extreme channels
and finding convex decompositions  into such channels. In particular,
a tight bound on the number of generalized extreme channels, i.e.,
channels which lie in the closure of the set of all extreme channels,
required for such a convex decomposition is not
known~\cite{Ruskai}. The set of extreme channels has been described by
Friedland and Loewy~\cite{Friedland} using the framework of semi-algebraic geometry. In contrast, we consider the set of extreme channels in the framework of differential geometry. In other words, this work focuses on assigning a smooth structure to the set of all extreme channels and we refer to~\cite{Friedland} for other interesting properties of this set. \\

The paper is structured as follows. First we give an overview of the
notation used in the paper. In Section~\ref{sec:mfd} we describe the
smooth manifold\footnote{We do not require a manifold to be connected.} structure of the set
$\mathcal{E}^{\textnormal{e}}_{\M,\N,\K}$ of $\M$ to $\N$ extreme
channels of a fixed Kraus rank $\K$: First, we adapt the
characterization of unital\footnote{A channel $\mathcal{E}$ is unital
  if $\mathcal{E}(I)=I$, where $I$ denotes the identity.} extreme
channels given by Choi~\cite{choi} to trace-preserving channels in
Section~\ref{sec:extreme_channels}. Then, in Section~\ref{sec:Kraus},
we describe the set of channels and extreme channels with the smooth
structure induced by the standard smooth structure on the Kraus
operators. In this picture, we find that
$\mathcal{E}^{\textnormal{e}}_{\M,\N,\K} \subset
\mathcal{E}_{\M,\N,\K}$
is an open subset and hence a smooth submanifold. In
Section~\ref{sec:Choi}, we transfer this topological property (founded
in the Kraus representation picture) to the Choi-state picture, which
will show that $\mathcal{E}^{\textnormal{e}}_{\M,\N,\K}$ can be
considered as a smooth submanifold of the set of all Choi-states of
fixed rank $\K$. In Section~\ref{sec:convex_comb}, we give a rigorous
proof of the known fact~\cite{Ruskai} that every channel can be
decomposed into a finite convex combination of extreme
channels. Finally, we look at an application to quantum information
theory in Section~\ref{sec:lower_bound}, where we derive a lower bound
on the number of parameters required for a quantum circuit topology
for extreme channels, which we have used in~\cite{Channels}.

 \section{Notation and Background}
 
 \subsection{Notation}
We use the notation $[A,B] \in \mathbb{C}^{\N\times (\M_A+\M_B)}$ to denote the (horizontal) concatenation of two matrices $A \in \mathbb{C}^{\N\times \M_A}$ and  $B \in \mathbb{C}^{\N\times\M_B}$, i.e., the first $\M_A$ columns of $[A,B]$ correspond to the columns of $A$ and the $(\M_A+1)$th column to the $(\M_A+\M_B)$th column to the columns of $B$.  And we denote the vertical concatenation of the matrices $A^T$ and $B^T$ by $[A^T;B^T]=[A,B]^T$. For arbitrary $\M,\N,\K \in \mathbb{N}$, we define:

\begin{itemize}
\item $\mathbb{C}^{\M \times \N}$: Complex $\M \times \N$ matrices
\item $H_{\M}$: Hermitian $\M \times \M$ matrices
\item $H_{\M,+}$: Positive semi-definite $\M \times \M$ matrices 
\item $H^{\K}_{\M,+}$: Elements in  $H_{\M,+}$ of rank $\K$
\item $V_{\M,\N}$: Set of all $V \in \mathbb{C}^{\N \times \M}$ s.t. $V^{\conjTrans}V=I$ (i.e., set of all isometries from an $\M$ to a $\N$ dimensional system) 
\item $V_{\M,\N,\K}$: Set of all $V=[A_1;A_2;\dots;A_{\K}] \in V_{\M,\K\N}$, such that the elements in  $\{A_i\}_{i \in \{1,2,\dots,\K\}} \in  \mathbb{C}^{\N \times \M}$  are linearly independent (over $\mathbb{C}$)
\item $U(\M)=V_{\M,\M}$: Unitary $\M \times \M$ matrices
\item $\mathcal{E}_{\M,\N}$: CPTP ($\mathbb{C}$-linear) maps from  $\mathbb{C}^{\M \times \M}$ to $\mathbb{C}^{\N \times \N}$
\item $\mathcal{E}^{\textnormal{v}}_{\M,\N}$: CP and unital ($\mathbb{C}$-linear) maps from  $\mathbb{C}^{\M \times \M}$ to $\mathbb{C}^{\N \times \N}$
\item $\mathcal{E}^{\textnormal{e}}_{\M,\N}$:  Elements in $\mathcal{E}_{\M,\N}$ that are extreme
\item $\mathcal{E}_{\M,\N,\K}$: Elements in $\mathcal{E}_{\M,\N}$ of Kraus rank $\K$ 
\item $\mathcal{E}_{\M,\N,\leqslant \K}=\bigcup_{j=1}^{\K}  \mathcal{E}_{\M,\N,j}$: Elements in $\mathcal{E}_{\M,\N}$ with Kraus rank at most $\K$ 
\item $\mathcal{E}^{\textnormal{e}}_{\M,\N,\K}= \mathcal{E}^{\textnormal{e}}_{\M,\N} \cap \mathcal{E}_{\M,\N,\K}$: Elements in $\mathcal{E}_{\M,\N,\K}$ that are extreme in $\mathcal{E}_{\M,\N}$
\item $\mathcal{C}_{\M,\N}$: Set of all Choi-states corresponding to channels from an $\M$ dimensional system $A$ to a $\N$ dimensional system $B$, i.e., $C_{AB}\in H_{\M\N,+}$, such that $\tr_B \,C_{AB}= \frac{1}{\M}I$
\item $\mathcal{C}^{\textnormal{e}}_{\M,\N}$: Elements in $\mathcal{C}_{\M,\N}$ that are extreme
\item $\mathcal{C}_{\M,\N,\K}$: Elements in $\mathcal{C}_{\M,\N}$ with rank $\K$
\item $ \mathcal{C}_{\M,\N,\leqslant \K}=\bigcup_{j=1}^{\K} \mathcal{C}_{\M,\N,j}$: Elements in  $\mathcal{C}_{\M,\N}$ with rank at most $\K$
\item $\mathcal{C}^{\textnormal{e}}_{\M,\N,\K}$: Elements in $\mathcal{C}_{\M,\N,\K}$ that are extreme in  $\mathcal{C}_{\M,\N}$
\end{itemize}

 \subsection{Restriction of the Domain or Image of a Smooth Map}
 
The following propositions give sufficient conditions for a map to remain smooth when its domain or image is restricted.

\begin{prop} [Theorem 5.27 of~\cite{Lee}] \label{prop:8.22}
Let $M$ and $N$ be smooth manifolds (with or without boundary). If $F:M \mapsto N$ is a smooth map and $D \subset M$ is an (immersed or embedded) submanifold, then  $\restr{F}{D}:D \mapsto N$ is smooth.
\end{prop}

\begin{prop}  [Corollary 5.30 of~\cite{Lee}]  \label{prop:8.25}
Let $M$ and $N$ be smooth manifolds, and $S \subset N$ be an embedded submanifold. Then any smooth map $F: M \mapsto N$ whose image is contained in $S$ is also smooth as a map from $M$ to $S$.

\end{prop}

\section{Smooth manifold structure for extreme channels} \label{sec:mfd}

\subsection{Characterization of Extreme Channels} \label{sec:extreme_channels}

We want to characterize the set of extreme points  $\mathcal{E}^{\textnormal{e}}_{\M,\N}\subset \mathcal{E}_{\M,\N}$. For this purpose we have to slightly modify Theorem 5 of~\cite{choi}. This modification was also considered in~\cite{Friedland}.

\begin{thm}[Characterization of extreme channels] \label{thm:choi}
Let $\mathcal{E} \in \mathcal{E}_{\M,\N}$ with minimal Kraus representation $\mathcal{E}(X)=\sum_{i=1}^{\K} A_i X A_i^{\conjTrans}$. Then $\mathcal{E}$ is extreme in $\mathcal{E}_{\M,\N}$ if and only if all elements of the set $\{A_i^{\conjTrans}A_j \}_{i,j \in \{1,2,\dots,\K\}}$ are linearly independent.

\end{thm}

\begin{rmk}
If the Kraus rank $\K$ of the channel $\mathcal{E}$ in Theorem~\ref{thm:choi} is bigger than $\M$, then $\mathcal{E}$ cannot be extreme, since in this case $|\{A_i^{\conjTrans}A_j \}_{i,j \in \{1,2,\dots,\K\}}|>\textnormal{dim}_{\mathbb{C}}(\mathbb{C}^{\M \times \M})$.
\end{rmk}

\begin{prop} \label{prop:1}
There exists a bijection $\Psi: \mathcal{E}_{\M,\N} \mapsto
\mathcal{E}^{\textnormal{v}}_{\N,\M}$ that sends extreme points of $ \mathcal{E}_{\M,\N} $ to extreme points of $ \mathcal{E}^{\textnormal{v}}_{\N,\M}$ and vice versa.

\end{prop}

\begin{proof}Note first that $\mathbb{C}^{l \times l}$ together with
  the Frobenius inner product is a (finite dimensional) Hilbert space
  for $l \in \mathbb{N}$. Let $\mathcal{E} \in
  \mathcal{E}_{\M,\N}$. Then $\mathcal{E}$ is $\mathbb{C}$-linear by
  definition and bounded. By the Fr\'echet-Riesz representation theorem there exists an injective map $\Psi(\mathcal{E})=\mathcal{E}^{\ast}: \mathcal{E}_{\M,\N} \mapsto \textnormal{B}(\mathbb{C}^{\N \times \N},\mathbb{C}^{\M \times \M})$, where $\textnormal{B}(\mathbb{C}^{\N \times \N},\mathbb{C}^{\M \times \M})$ denotes the set of linear bounded operators from $\mathbb{C}^{\N \times \N}$ to $\mathbb{C}^{\M \times \M}$ and $\mathcal{E}^{\ast}$ denotes the adjoint map of $\mathcal{E}$, i.e., for all $C \in \mathbb{C}^{\M \times \M}$ and  $D \in \mathbb{C}^{\N \times \N}$ we have $\left\langle \mathcal{E}(C),D\right\rangle=\left\langle C,\mathcal{E}^{\ast}(D)\right\rangle$. Let $\{A_i\}_{i \in \{1,2,\dots,\K\}}$ be the Kraus operators of $\mathcal{E}$. By a direct computation one finds that the Kraus operators of $\mathcal{E}^{\ast}$ are $\{A_i^{\conjTrans}\}_{i \in \{1,2,\dots,\K\}}$, and therefore, $\mathcal{E}^{\ast} \in  \mathcal{E}^{\textnormal{v}}_{\N,\M}$. Since  $(\mathcal{E}^{\ast})^{\ast}=\mathcal{E}$, we can set $C=I$ in the adjoint property above to see that $\Psi^{-1}$ sends unital maps  to trace-preserving maps. Therefore the map $\Psi: \mathcal{E}_{\M,\N} \mapsto \mathcal{E}^{\textnormal{v}}_{\N,\M}$ is a bijection.\\
  Assume that $\mathcal{E} \in \mathcal{E}_{\M,\N}$ is not extreme,
  i.e., there exist $\mathcal{E}_1,\mathcal{E}_2 \in
  \mathcal{E}_{\M,\N}$, $\mathcal{E}_1 \neq \mathcal{E} $,
  $\mathcal{E}_2 \neq \mathcal{E} $ and $p\in (0,1)$ s.t.\
  $\mathcal{E}=p\mathcal{E}_1+(1-p)\mathcal{E}_2$. By the linearity of
  the adjoint map, we have
  $\mathcal{E}^{\ast}=p\mathcal{E}^{\ast}_1+(1-p)\mathcal{E}^{\ast}_2$,
  which shows that elements of $\mathcal{E}_{\M,\N}$ that are not
  extreme cannot be mapped to extreme elements of
  $\mathcal{E}^{\textnormal{v}}_{\N,\M}$. The reverse direction
  follows analogously.
\end{proof}

\begin{proof}[Proof of Theorem~\ref{thm:choi}]
  Theorem~5 of~\cite{choi} shows that $\mathcal{E}^{\ast} \in
  \mathcal{E}^{\textnormal{v}}_{\N,\M}$ (with linearly independent
  Kraus operators $\tilde{A}_i \in \mathbb{C}^{\M \times \N}$) is
  extreme if and only if the elements in $\{\tilde{A}_i\tilde{A}_j
  ^{\conjTrans}\}_{i,j \in \{1,2,\dots,\K\}}$ are linearly
  independent. By Proposition~\ref{prop:1}, this leads to a
  characterization of the extreme points in
  $\mathcal{E}_{\M,\N}$. Since the Kraus operators of $\mathcal{E}$
  (where $\mathcal{E}$ denotes the adjoint map of
  $\mathcal{E}^{\ast}$) are $A_i:=\tilde{A}_i^{\conjTrans} \in
  \mathbb{C}^{\N \times \M}$ (cf.\ the proof of
  Proposition~\ref{prop:1}), the map $\mathcal{E} \in
  \mathcal{E}_{\M,\N}$ is extreme if and only if the elements in
  $\{A^{\conjTrans}_iA_j \}_{i,j \in \{1,2,\dots,\K\}}$ are linearly
  independent.
\end{proof}

\subsection{Structure of the Set of Extreme Channels in the Kraus Representation} \label{sec:Kraus}
 
In this section we consider the smooth structure of the set
$\mathcal{E}^{\textnormal{e}}_{\M,\N,\K}$ working with the Kraus
representation of channels. Our first goal is to describe the set
$\mathcal{E}_{\M,\N,\K}$ of $\M$ to $\N$ channels with Kraus rank $\K$. We can assume that $\M\leqslant \K\N$ and $\K\leqslant \M\N$, since $\mathcal{E}_{\M,\N,\K}=\emptyset$ if $\M>\K\N$ (cf.\ Lemma~6 of~\cite{Friedland}) or $\K>\M\N$. Let $\{A_i\}_{i \in \{1,2,\dots,\K\}}$ denote a set of (linearly independent) Kraus operators of $\mathcal{E} \in \mathcal{E}_{\M,\N,\K}$. Then we define $V=[A_1;A_2;\dots;A_{\K}]$ which lies in $V_{\M,\N,\K}$, because $V^{\conjTrans}V=\sum_{i=1}^{\K} A_i^{\conjTrans}A_i=I$. Since the Kraus representation is not unique, we do not have a one-to-one correspondence between $\mathcal{E}_{\M,\N,\K}$ and $V_{\M,\N,\K}$. However, we can exploit the desired correspondence by taking the quotient of  $V_{\M,\N,\K}$ with respect to the unitary freedom of the Kraus operators.\\

In the following, we always assume that $\M\leqslant \K\N$ and
$\K\leqslant \M\N$. The next Lemma is generally known (see for example~\cite{Fuchs(Stiefel manifold)}).

\begin{lem} [Stiefel manifold] \label{lem:Stiefel}
Let $\M\leqslant \N$. Then the Stiefel manifold $V_{\M,\N}$ is a compact, smooth embedded submanifold of $\mathbb{R}^{2 \N \M}$ of dimension $2\M\N-\M^2$.
\end{lem}

\begin{prop}  \label{prop:V_M,N,K_is_manifold}
The set $V_{\M,\N,\K}$ is an open subset of $V_{\M,\K\N}$. In particular, $V_{\M,\N,\K}$ is a smooth embedded submanifold of $V_{\M,\K\N}$.

\end{prop}
\begin{proof}
  We can write all coefficients of a complex $\N\times \M$ matrix in a
  column vector leading to a natural correspondence $\psi:
  \mathbb{C}^{\N\times \M} \mapsto \mathbb{C}^{\N\M}$. Let
  $l={{{\N\M}}\choose{\K}}$. We define the map
  $F:V_{\M,\K\N}\rightarrow \mathbb{C}^{l}$ sending
  $V=[A_1;A_2;\dots;A_{\K}]$ to all $\K\times \K$ minors\footnote{An
    $\K\times \K$ minor of a matrix $D$ is the determinant of an
    $\K\times \K$ sub-matrix of $D$ formed by `deleting' certain rows
    (or columns).} of the matrix
  $[\psi(A_1),\psi(A_2),\dots,\psi(A_{\K})]$ (ordered in an arbitrary
  way). Then the condition that the elements of the set $\{A_i \}_{i
    \in \{1,2,\dots,\K\}}$ are linearly independent reads: $F(V) \neq
  (0,0,\dots,0)$.  Since $F$ is continuous,
  $F^{-1}(\{(0,0,\dots,0)\}^{\textnormal{c}})=V_{\M,\N,\K}$ is open.
\end{proof}

We can use Theorem~21.10 of~\cite{Lee} to describe the manifold structure of the orbit space  $V_{\M,\N,\K}/U(\K)$.

\begin{defi} A group $G$ acts \emph{freely}
  on a set $S$ if the only element of $G$ that fixes any element of
  $S$ is the identity, i.e., for all $p \in S$ and $g \in G$,
  $g\cdot p=p$ implies $g=I$.
\end{defi}

\begin{defi} Let $G$ be a Lie group that acts
  continuously on a manifold $M$. The action is said to be
  \emph{proper} if the map $G\times M \mapsto M \times M$ given by
  $(g,p)\mapsto (g\cdot p, p)$ is a proper map, i.e., the preimage of
  a compact set is compact.
\end{defi}
 
 The following Proposition gives a sufficient condition for a group action to be proper.
 
\begin{prop} [Corollary~21.6 of~\cite{Lee}] \label{prop:9.14}
Any continuous action by a compact Lie group on a manifold is proper.
\end{prop}

\begin{thm} [Quotient Manifold Theorem~\cite{Lee}] \label{thm:quotient}
Suppose a Lie group $G$ acts smoothly, freely, and properly on a smooth manifold $M$. Then the orbit space $M/G$ is a topological manifold of dimension equal to $\dim(M) - \dim(G)$, and has a unique smooth structure with the property that the quotient map $\pi: M \mapsto M/G$ is a smooth submersion.
\end{thm} 

\begin{lem} [Lemma~21.1 of~\cite{Lee}] \label{lem:9.15}
For any continuous action of a topological group $G$ on a topological space $M$, the quotient map $\pi: M \mapsto M / G$ is open.
\end{lem}

\begin{prop}  \label{prop:action_U(\K)}
The Lie group $U(\K)$ acts smoothly, freely and properly on the
manifold $V_{\M,\N,\K}$ by the action $U\cdot V=(U\otimes I)V$, where
$U \in U(\K)$, $V \in V_{\M,\N,\K}$ and  $I$ denotes the $\N \times
\N$ identity matrix.
\end{prop}

\begin{proof}
  We first show that $U\cdot V \in V_{\M,\N,\K}$ for all $U\in U(\K)$ and
  $V\in V_{\M,\N,\K}$. Note that
  $V^{\conjTrans}(U^{\conjTrans}\otimes I)(U\otimes I)V=I$ and that the
  linear independence of the matrices $A_i$ (where
  $V=[A_1;A_2;\dots;A_{\K}]$) is preserved under the unitary action:
  Assume
  $\sum_{i=1}^{\K} \alpha_i \left(\sum_{j=1}^{\K} (U)_{ij}A_j\right)=0$ for
  some coefficients $\alpha_i \in \mathbb{C}$. This is equivalent to
  $\sum_{j=1}^{\K} \left(\sum_{i=1}^{\K} \alpha_i (U)_{ij}\right)A_j=0$
  which implies $\sum_{i=1}^{\K} \alpha_i (U)_{ij}=0$ for all
  $j\in \{1,2, \dots, \K\}$, since the $A_j$ are linearly
  independent. Since $U$ is unitary, this implies $\alpha_i=0$ for
  all $i\in \{1,2,\dots,\K\}$. We conclude that the group action is
  well defined. 
  To show that the action is free, choose a $V \in V_{\M,\N,\K}$ and a
  $U \in U(\K)$ and assume that $U \cdot V=V$. Writing
  $V=[A_1;A_2;\dots;A_{\K}]$, the last equation becomes
  $\sum_{j=1}^{\K}(U)_{ij}A_j=A_i$ for all $i \in \{1,2,\dots,\K\}$.
  Since $V \in V_{\M,\N,\K}$, the $A_j$ are linearly independent and we
  conclude that $(U)_{ij}=\delta_{ij}$ or equivalently $U=I$. To see
  that the action is smooth, consider the map
  $F\left(U,V\right)=(U\otimes I) V$:
  $\mathbb{C}^{\K \times \K} \times \mathbb{C}^{\K\N \times \M} \rightarrow
  \mathbb{C}^{\K\N \times \M}$.\footnote{We
    always identify $\mathbb{C}\cong\mathbb{R}^2$, and hence we can treat $F$
    as a map from $\mathbb{R}^{2\left(\K^2+\K\N\M\right)}$ to
    $\mathbb{R}^{2\K\N\M}$.} Since taking tensor products is a smooth
  operation, the map $F$ is smooth. By Propositions~\ref{prop:8.22}
  and~\ref{prop:V_M,N,K_is_manifold} and Lemma~\ref{lem:Stiefel}, the map
  $\tilde{F}\left(U,V\right)=(U \otimes I)V$:
  $U(\K) \times V_{\M,\N,\K} \mapsto \mathbb{C}^{\K\N \times \M}$ is
  smooth. Then, by Propositions~\ref{prop:8.25}
  and~\ref{prop:V_M,N,K_is_manifold} and Lemma~\ref{lem:Stiefel}, the map
  $F'\left(U,V\right)=(U \otimes I)V$:
  $U(\K) \times V_{\M,\N,\K} \mapsto V_{\M,\N,\K}$ is also smooth.  Since
  the Lie group $U(\K)$ is compact, the action is proper by
  Proposition~\ref{prop:9.14}.
\end{proof}

\begin{defi}
  We define the equivalence relation $\sim$ as follows: Let $V_1,V_2
  \in V_{\M,\N,\K}$. Then $V_1\sim V_2$ if there exists a $U\in
  U(\K)$, such that $U\cdot V_1=V_2$. The orbit space is
  $V_{\M,\N,\K}/U(\K):=\{[V]:V \in V_{\M,\N,\K}\}$ (together with the
  quotient topology).
\end{defi}

\begin{lem} \label{lem:Quotient} The orbit space $V_{\M,\N,\K}/U(\K)$
  is a topological manifold of dimension equal to
  $\dim(V_{\M,\N,\K})-\dim(U(\K))=2\M\K\N-\M^2-\K^2$ with a unique
  smooth structure such that the quotient map $\pi:V_{\M,\N,\K}
  \mapsto V_{\M,\N,\K}/U(\K)$ is a smooth submersion. Moreover, $\pi$
  is an open map.
\end{lem}
\begin{proof}The first part of the theorem follows from
  Theorem~\ref{thm:quotient}, where the assumption for the theorem are
  satisfied because of Proposition~\ref{prop:action_U(\K)}. The
  quotient map $\pi$ is open by Lemma~\ref{lem:9.15}.
\end{proof}

\begin{lem} \label{lem:channels_isom} There is a one-to-one
  correspondence between the set $\mathcal{E}_{\M,\N,\K}$ of channels
  of Kraus rank $\K$ and the orbit space $V_{\M,\N,\K}/U(\K)$.
\end{lem}

\begin{proof}We define the quotient map $\pi(V)=[V]$: $V_{\M,\N,\K}
  \mapsto V_{\M,\N,\K}/U(\K)$ and the map $\psi:
  \mathcal{E}_{\M,\N,\K} \mapsto V_{\M,\N,\K}/U(\K)$, by sending a
  channel $\mathcal{E} \in \mathcal{E}_{\M,\N,\K}$ with (linearly
  independent) Kraus operators $\{A_i\}_{i \in \{1,2,\dots,\K\}}$ to
  $\pi\left([A_1;A_2;\dots;A_{\K}]\right)$. To show that this map is
  well defined, we must show that it is independent on the choice of
  the Kraus operators. By Remark 4 of~\cite{choi}, two Kraus
  representations $\{A_i\}_{i \in \{1,2,\dots,\K\}}$ and $\{B_i\}_{i
    \in \{1,2,\dots,\K\}}$ describe the same channel $\mathcal{E} \in
  \mathcal{E}_{\M,\N,\K}$ if and only if there exist a unitary $U \in
  U(\K)$, such that $B_j=\sum_{i=1}^{\K}(U)_{ji}A_i$ for all $j \in
  \{1,2,\dots,\K\}$ or equivalently $V_B=(U\otimes I )V_A $, where
  $V_A=[A_1;A_2;\dots;A_{\K}] \in V_{\M,\N,\K}$ and
  $V_B=[B_1;B_2;\dots;B_{\K}] \in V_{\M,\N,\K}$. Therefore, the two
  Kraus representation describe the same channel if and only if
  $V_A\sim V_B$. We conclude that the map $\psi$ is well defined and
  injective. On the other hand, for all $W\in V_{\M,\N,\K}/U(\K)$, we
  can define a channel $\mathcal{E}\in\mathcal{E}_{\M,\N,\K}$ with
  $\psi(\mathcal{E})=W$ by choosing a representative element $V\in
  \pi^{-1}(W)$, breaking it into blocks $V=[A_1;A_2;\dots;A_{\K}]$ and
  treating those as the channel's Kraus operators. This shows that
  $\psi$ is also surjective.
\end{proof}

We are now ready to study the structure of the set
$\mathcal{E}^e_{\M,\N,\K}$ of extreme channels. In the following we
show that we can identify $\mathcal{E}^e_{\M,\N,\K}$ with a smooth
manifold.

\begin{prop} \label{prop:open} The set
  $\tilde{O}:=\{[A_1;A_2;\dots;A_{\K}] \in V_{\M,\N,\K}:
  \{A_i^{\conjTrans}A_j \}_{i,j \in \{1,2,\dots,\K\}}$ are linearly
  independent$\}$ is an open subset of the manifold $V_{\M,\N,\K}$.
\end{prop}

\begin{proof}Works analogously to the proof of Proposition~\ref{prop:V_M,N,K_is_manifold}.
% We can write all coefficients of a complex matrix in a column vector leading to a natural correspondence $\psi: \mathbb{C}^{2^n\times 2^m} \mapsto \mathbb{C}^{2^{n+m}}$. Let $N={{2^{nk}}\choose{2^m}}$ if $m\leqslant n+k$ and  $N={{2^{m}}\choose%{2^{n+m}}}$ if $m>n+k$. We define the map $F:V_{\M,\K\N}\mapsto \mathbb{C}^{N}$  sending $V=[A_1;A_2;\dots;A_{2^k}]$ to all minors of the matrix $[\psi(A_1^{\conjTrans}A_1),\dots,\psi(A_{2^k}^{\conjTrans}A_{2^k})]$ (ordered in a arbitrary way). Then the condition that the elements of $\{A_i^{\conjTrans}A_j \}_{i,j \in \{1,2,\dots,\K\}}$ are linearly independent reads: $F(V) \neq (0,0,\dots,0)$. Since $F$ is continuous, we are done. 
\end{proof}

\begin{lem} [Manifold structure for $\mathcal{E}^e_{\M,\N,\K}$] \label{lem:Extreme_channels}
Let $O:=\pi(\tilde{O})\subseteq V_{\M,\N,\K}/U(\K)$, where $\pi$ is the quotient map of
Lemma~\ref{lem:Quotient}.  $O$ is a smooth manifold of dimension
$2\M\K\N-\K^2-\M^2$ and there is a one-to-one correspondence between
the set $\mathcal{E}^e_{\M,\N,\K}$ of extreme channels of Kraus rank
$\K$ and $O$. Moreover, $\mathcal{E}^e_{\M,\N,\M}\neq
\emptyset$.
\end{lem}

\begin{proof}
  Since $\pi$ is an open map and $\tilde{O}$ is an open subset of
  $V_{\M,\N,\K}$, $O$ is an open subset of the orbit space $V_{\M,\N,\K}/U(\K)$. Together with Lemma~\ref{lem:channels_isom} and Theorem~\ref{thm:choi}, this implies the first part of the Lemma. \\
  For the second part, we borrow an argument
  from~\cite{Friedland}. Consider a channel $\mathcal{E}$ with Kraus
  operators $A_i=\ketbra{\psi}{i}$ for $i \in \{1,2,\dots,\M \}$,
  where $\ket{\psi}\in \mathbb{C}^{\N}$ is of unit length and
  $\ket{i}\in \mathbb{C}^{\M}$ denotes the $i$th standard basis
  vector.\footnote{Note that $\bra{\phi} \in \mathbb{C}^{1 \times d} $
    denotes the conjugate transpose of a $d$-dimensional vector
    $\ket{\phi} \in \mathbb{C}^{d \times 1}$.}, Note that
  $\sum_{i=1}^{\M} A_i^{\conjTrans}A_i=I$ and that the elements in the
  set
  $\{A_i^{\conjTrans}A_j \}_{i,j \in \{1,2,\dots,\M\}}=\{\ketbra{i}{j}
  \}_{i,j \in \{1,2,\dots,\M\}}$
  are linearly independent. By Theorem~\ref{thm:choi}, we have
  $\mathcal{E} \in \mathcal{E}^e_{\M,\N,\M}$ and therefore
  $\mathcal{E}^e_{\M,\N,\M}\neq \emptyset$.
\end{proof}
Note that the above shows that the channel corresponding to the
operation of discarding a system (tracing out) and then generating a
new pure state is extremal.
 
 \subsection{Structure of the Set of Extreme Channels in the Choi-State Representation} \label{sec:Choi}
 
 We found a smooth description of the set of extreme channels $\mathcal{E}^e_{\M,\N,\K}$ in Section~\ref{sec:Kraus}. This will allow us to transfer the characterization of extreme channels to the Choi-state representation. 

\begin{lem} [Manifold structure for  $H_{\M,+}^{\K}$] \label{lem:H_{\M,+}^{\K}}
The set  $H_{\M,+}^{\K}$ is a smooth embedded submanifold of $\mathbb{R}^{2\M^2}$ of dimension $2\M\K-\K^2$ .
\end{lem}

The `real case' of Lemma~\ref{lem:H_{\M,+}^{\K}} was shown
in~\cite{Helmke} (cf.\ also~\cite{Vandereyecken}), where they
considered the manifold of real symmetrical $\M \times \M$ matrices of
rank $\K$. Our proof is a straightforward generalization of the proof given in~\cite{Helmke} to the complex case. We begin with some
preparatory results.

\begin{defi}
We define $E_{\K}$ to be a square matrix,\footnote{The dimension of the matrix will always be clear from the context.} whose first $\K$ diagonal entries are equal to one and all the other entries are equal to zero. 
\end{defi}

\begin{prop}  \label{prop:Orbit_description_of_H_{\M,+}^{\K}}
 We have $H_{\M,+}^{\K}=\{AE_{\K}A^{\conjTrans}:A \in \textnormal{GL}(\mathbb{C},\M)\}$.
\end{prop}

\begin{proof}
  The inclusion ``$\supseteq$'' is clear. To see the inclusion
  ``$\subseteq$'', let $H \in H_{\M,+}^{\K}$. By the spectral theorem
  there exists a $U \in U(\M)$ such that $H=UDU^{\conjTrans}$, where
  $D$ is a $\M \times \M$ matrix with positive diagonal entries
  $d_1,d_2,\dots,d_{\K}$ and zeros elsewhere. We define $\tilde{D}$ as
  the matrix $D$ where we replace the zeroes on the diagonal by
  ones. Then we have $D=\sqrt{\tilde{D}}E_{\K}\sqrt{\tilde{D}}$, and
  hence
  $H=U\sqrt{\tilde{D}}E_{\K}\sqrt{\tilde{D}}U^{\conjTrans}=AE_{\K}A^{\conjTrans}$,
  where we set $A=U\sqrt{\tilde{D}} \in
  \textnormal{GL}(\mathbb{C},\M)$.
\end{proof}

A sufficient condition for orbits of Lie group actions to be smooth
manifolds was given in~\cite{Gibson}.

\begin{defi}
  A map $f:D\rightarrow\mathbb{R}^{\N}$ with $D \subset \mathbb{R}^{\M}$ is
  \emph{semialgebraic} if the graph of $f$ is semialgebraic in
  $\mathbb{R}^{\M} \times \mathbb{R}^{\N}$.
\end{defi}

\begin{thm} [Theorem~B4 of Appendix~B
  of~\cite{Gibson}] \label{thm:Gibson} Let $\Phi: G \times S \mapsto
  S$ be a smooth action of a Lie group G on a smooth manifold S. And
  suppose that the action is semialgebraic. Then all the orbits are
  smooth embedded submanifolds of $S$.\footnote{Note that ``smooth
    submanifolds'' in Theorem~B4 of Appendix~B of~\cite{Gibson}
    corresponds to a ``smooth embedded submanifold'' in our
    terminology.}
\end{thm}

\begin{proof}[Proof of Lemma~\ref{lem:H_{\M,+}^{\K}} (part 1)]
  We define the map
  $\Phi(A,H)=AHA^{\conjTrans}:\textnormal{GL}(\mathbb{C},\M) \times
  \mathbb{C}^{\M \times \M} \mapsto \mathbb{C}^{\M \times \M}$. Note
  that $\Phi$ describes a smooth action of the Lie group
  $\textnormal{GL}(\mathbb{C},\M)$ on the smooth manifold
  $\mathbb{C}^{\M \times \M}$. Moreover, $\Phi$ is a semialgebraic
  map: The (complex) graph of $\Phi$ corresponds to the set
  $\{(A,H,\tilde{H})\in (\mathbb{C}^{\M \times \M})^{\times 3}:
  \textnormal{det}(A)\neq 0 \textnormal{ and }
  AHA^{\conjTrans}-\tilde{H}=0\}$. We can embed the complex space
  $(\mathbb{C}^{\M \times \M})^{\times 3}$ into $(\mathbb{R}^{2\M
    \times \M})^{\times 3}$ and rewrite the conditions as real
  polynomial equations.  By Theorem~\ref{thm:Gibson} and
  Proposition~\ref{prop:Orbit_description_of_H_{\M,+}^{\K}}, we
  conclude that the orbit $H_{\M,+}^{\K}=\{AE_{\K}A^{\conjTrans}:A \in
  \textnormal{GL}(\mathbb{C},\M)\}$ is a smooth embedded submanifold
  of $\mathbb{R}^{2\M \times \M}\cong \mathbb{R}^{2\M^2}$. To
  determine the dimension of this manifold, we need an additional
  result.
\end{proof}

\begin{prop} \label{prop:Tangent_space_of_H_{\M,+}^{\K}} Let
  $p \in H_{\M,+}^{\K}$ and write $p=A_pE_{\K}A_p^{\conjTrans}$ for some
  $A_p \in \textnormal{GL}(\mathbb{C},\M)$. Then the tangent space at
  $p$ is given by
  $\textnormal{T}_p H_{\M,+}^{\K}=\{\Delta E_{\K}
  A_p^{\conjTrans}+A_pE_{\K}\Delta^{\conjTrans}: \Delta \in \mathbb{C}^{\M
    \times \M} \}$.
\end{prop}

\begin{proof}
%"$\supseteq$": Let  $p \in H_{\M,+}^{\K}$ as in the claim, choose a $\Delta \in \mathbb{C}^{\M \times \M}$ and define $X:=\Delta E_{\K} A_p^{\conjTrans}+A_pE_{\K}\Delta^{\conjTrans}$. Since $\textnormal{T}_{A_p} \textnormal{GL}(\mathbb{C},\M)= \mathbb{C}^{\M \times \M}$ we can find a smooth path $\gamma:(-\epsilon,\epsilon) \mapsto \textnormal{GL}(\mathbb{C},\M)$ with $\epsilon \in \mathbb{R}$ such that $\gamma(0)=A_p$ and $\gamma'(0):=\left. \frac{\textnormal{d}}{\textnormal{dt}} \right|_{t=0}  \gamma(t)=\Delta$. By Proposition~\ref{prop:Orbit_description_of_H_{\M,+}^{\K}}, this induces a smooth path in $H_{\M,+}^{\K}$ by setting $\tilde{\gamma}(t)=\gamma(t)E_{\K}\gamma(t)$. Then  $\tilde{\gamma}'(t)=\gamma'(0)E_{\K}\gamma(0)^{\conjTrans}+\gamma(0)E_{\K}\gamma'(0)^{\conjTrans}=\Delta E_{\K}A_p^{\conjTrans}+A_pE_{\K} \Delta^{\conjTrans}$ and $\tilde{\gamma}(0)=A_pE_{\K}A_p^{\conjTrans}=p$. We conclude that  $\tilde{\gamma}'(0)=\Delta E_{\K}A_p^{\conjTrans}+A_pE_{\K} \Delta^{\conjTrans} \in \textnormal{T}_p H_{\M,+}^{\K}$.\\
%Let $p=A_pE_{\K}A_p^{\conjTrans} \in H_{\M,+}^{\K}$ for some $A_p \in
%\textnormal{GL}(\mathbb{C},\M)$. 
We define the map
$\phi(A)=AE_{\K}A^{\conjTrans}:\textnormal{GL}(\mathbb{C},\M) \rightarrow
H_{\M,+}^{\K}$ (where we used
Proposition~\ref{prop:Orbit_description_of_H_{\M,+}^{\K}} to determine the
image space). Note that the map $\phi$ is smooth, since from the first
part of the proof of Lemma~\ref{lem:H_{\M,+}^{\K}}, the set
$H_{\M,+}^{\K}$ is a smooth embedded submanifold (and hence we can
apply Proposition~\ref{prop:8.25} to the smooth map
$\phi'(A)=AE_{\K}A^{\conjTrans}:\textnormal{GL}(\mathbb{C},\M) \mapsto
\mathbb{C}^{\M \times \M}$). Then the pushforward of $\phi$ at $A_p$
is given by $\textnormal{D}\phi_{A_p}\left(\Delta\right)=\Delta
E_{\K}A_p^{\conjTrans}+A_pE_{\K}
\Delta^{\conjTrans}:\textnormal{T}_{A_p}\textnormal{GL}(\mathbb{C},\M)
\mapsto \textnormal{T}_p H_{\M,+}^{\K}$. The inclusion ``$\supseteq$''
of the claim in Proposition~\ref{prop:Tangent_space_of_H_{\M,+}^{\K}}
follows, because  $\textnormal{T}_{A_p}
\textnormal{GL}(\mathbb{C},\M)\cong \mathbb{C}^{\M \times \M}$. To see
the inclusion ``$\subseteq$'', we show that $\phi$ has constant
rank. To see this, note that the pushforward of $\phi$ at an arbitrary
$A \in \textnormal{GL}(\mathbb{C},\M)$ is related to the pushforward
at the identity in the following way
$\textnormal{D}\phi_{A}\left(\Delta A\right)=\Delta A
E_{\K}A^{\conjTrans}+AE_{\K}
A^{\conjTrans}\Delta^{\conjTrans}=A\left(A^{-1} \Delta A
  E_{\K}+E_{\K}A^{\conjTrans}
  \Delta^{\conjTrans}\left(A^{\conjTrans}\right)^{-1}
\right)A^{\conjTrans}=A\,\textnormal{D}\phi_{I}\left( A^{-1}\Delta
  A\right)A^{\conjTrans}$. This implies that $\tilde{\Delta}  \in
\textnormal{ker}\left(\textnormal{D}\phi_{A}\right)$ if and only if
$A^{-1}\tilde{\Delta}  \in
\textnormal{ker}\left(\textnormal{D}\phi_{I}\right)$ and hence
$\textnormal{dim}\left(\textnormal{ker}\left(\textnormal{D}\phi_{I}\right)\right)=\textnormal{dim}\left(\textnormal{ker}\left(\textnormal{D}\phi_{A}\right)\right)$
for all  $A \in \textnormal{GL}(\mathbb{C},\M)$. By the rank-nullity
theorem, we conclude that
$\textnormal{rank}\left(\textnormal{D}\phi_{I}\right)=\textnormal{rank}\left(\textnormal{D}\phi_{A}\right)$
for all  $A \in \textnormal{GL}(\mathbb{C},\M)$, i.e., $\phi$ has
constant rank. Since $\phi$ is also surjective by
Proposition~\ref{prop:Orbit_description_of_H_{\M,+}^{\K}}, we can
apply the global rank theorem (cf. for example Theorem~4.14
of~\cite{Lee}) to see that $\phi$ is a submersion. In particular, $\textnormal{D}\phi_{A_p}$ is surjective, which shows the inclusion  ``$\subseteq$''.
\end{proof}

We are now ready to determine the dimension of the manifold $H_{\M,+}^{\K}$.\\

\begin{proof}[Proof of Lemma~\ref{lem:H_{\M,+}^{\K}} (part 2)]
  We define the map $\phi$ as in the proof of
  Proposition~\ref{prop:Tangent_space_of_H_{\M,+}^{\K}}. Then
  $\textnormal{D}\phi_{I}(\Delta)=\Delta E_{\K}+E_{\K}
  \Delta^{\conjTrans}$. Since $\phi$ is a submersion (cf.\ proof of
  Proposition~\ref{prop:Tangent_space_of_H_{\M,+}^{\K}}),
  $\textnormal{D}\phi_{I}$ is surjective and hence, by the
  rank-nullity theorem,
  $\textnormal{dim}\left(\textnormal{T}_{E_{\K}}H_{\M,+}^{\K}\right)=\textnormal{dim}\left(\textnormal{image}\left(\textnormal{D}\phi_{I}\right)\right)=\textnormal{dim}\left(\textnormal{T}_{I}
    \textnormal{GL}(\mathbb{C},\M)
  \right)-\textnormal{dim}\left(\textnormal{ker}\left(\textnormal{D}\phi_{I}\right)\right)$,
  where $\textnormal{dim}\left(\textnormal{T}_{I}
    \textnormal{GL}(\mathbb{C},\M) \right)=2\M^2$. Note that $\Delta
  \in \textnormal{ker}\left(\textnormal{D}\phi_{I}\right)$ if and only
  if $\Delta E_{\K}=-E_{\K} \Delta^{\conjTrans}$. Writing $\Delta$ in
  block matrix form
  $\Delta=[\Delta_{11},\Delta_{1,2};\Delta_{21},\Delta_{22}]$, where
  $\Delta_{11} \in \mathbb{C}^{\K \times \K}$, the condition above is
  equivalent to the two conditions $\Delta_{21}=0$ and
  $\Delta_{11}=-\Delta_{11}^{\conjTrans}$. Therefore
  $\textnormal{dim}\left(\textnormal{ker}\left(\textnormal{D}\phi_{I}\right)\right)=2\M(\M-\K)+\K^2$
  and hence
  $\textnormal{dim}\left(\textnormal{T}_{E_{\K}}H_{\M,+}^{\K}\right)=2\M^2-\left(2\M(\M-\K)+\K^2\right)=2\M\K-\K^2$.
\end{proof}

Lemma~\ref{lem:H_{\M,+}^{\K}} allows us to show that the set of all
Choi-states corresponding to channels from an $\M$-dimensional to a
$\N$-dimensional system of Kraus rank $\K$ is a smooth manifold.

\begin{lem} [Manifold structure for
  $\mathcal{C}_{\M,\N,\K}$] \label{lem:C_{M,N,K}} The set
  $\mathcal{C}_{\M,\N,\K}$ is a smooth embedded submanifold of
  $\mathbb{R}^{2\M^2\N^2}$. Its dimension is $2\M\K\N-\K^2-\M^2$.
\end{lem}

\begin{proof}
  Define the smooth map $\Psi(H_{AB})=\tr_B \, H_{AB}:H_{\M\N,+}^{\K}
  \mapsto H_{\M}$ (the smoothness follows again from
  Proposition~\ref{prop:8.22} and~\ref{prop:8.25}). By the Regular
  Level Set Theorem (cf.\ Corollary~5.14 of~\cite{Lee}) and because
  $\textnormal{dim}(H_{\M})=\M^2$ and
  $\textnormal{dim}(H_{\M\N,+}^{\K})=2\M\K\N-\K^2$ (cf.\
  Lemma~\ref{lem:H_{\M,+}^{\K}}), it suffices to show that
  $p':=\frac{1}{\M}I \in H_{\M}$ is a regular value of $\Psi$, i.e.,
  that for all $p \in \Psi^{-1}(p')$ the pushforward
  $D\Psi_p:\textnormal{T}_pH_{\M\N,+}^{\K} \mapsto
  \textnormal{T}_{p'}H_{\M}$ is surjective. To see this, choose $p \in
  \Psi^{-1}(p')$ and write $p=A_pE_{\K}A_p^{\conjTrans}$ for some $A_p
  \in \textnormal{GL}(\mathbb{C},\M\N)$. Choose a tangent vector $X'
  \in \textnormal{T}_{p'}H_{\M}\cong H_{\M}$ and write
  $X'=C+C^{\conjTrans}$, where $C=\frac{1}{2}X' \in \mathbb{C}^{\M
    \times \M}$. Define $\Delta:=\M(C \otimes I)A_p \in
  \mathbb{C}^{\M\N \times \M\N}$ and $X=\Delta
  E_{\K}A_p^{\conjTrans}+A_pE_{\K}\Delta^{\conjTrans} \in
  \textnormal{T}_pH_{\M\N,+}^{\K}$ (by
  Proposition~\ref{prop:Tangent_space_of_H_{\M,+}^{\K}}). Since the
  partial trace is a linear (and continuous) map, we have:
  $D\Psi_p(X)=\Psi(X)=\tr_B \, \Delta
  E_{\K}A_p^{\conjTrans}+\left(\tr_B\, \Delta
    E_{\K}A_p^{\conjTrans}\right)^{\conjTrans}$. Using the definition
  of $\Delta$, we have $\tr_B\, \Delta E_{\K}A_p^{\conjTrans}=\M \,
  \tr_B\, (C\otimes I)A_p E_{\K}A_p^{\conjTrans}=\M C\,\tr_B \, A_p
  E_{\K}A_p^{\conjTrans}=\M C\Psi(p)=\M Cp'=C$. We conclude that
  $D\Psi_p(X)=C+C^{\conjTrans}=X'$. Since $X'\in
  \textnormal{T}_{p'}H_{\M}$ was arbitrary, we showed that the
  pushforward $D\Psi_p$ is surjective for any $p \in \Psi^{-1}(p')$.
\end{proof}

To describe the set of extreme channels in the Choi-state
representation, we transfer the description of
Lemma~\ref{lem:Extreme_channels} to $\mathcal{C}_{\M,\N,\K}$.

\begin{defi} \label{defi:Choi_Iso} We use
  $\ket{\gamma}_{A'A}=\frac{1}{\sqrt{s}} \sum_i\ket{i}_{A'}\ot\ket{i}_A \in
  \mathbb{C}^{\M^2}$ to denote the maximally entangled
  state between the $\M$-dimensional system $A$ and a copy of this
  system, denoted by $A'$. We define the \emph{Choi
    map} $\Gamma(\mathcal{E})=\mathcal{I}_{A'}\ot\mathcal{E}(\proj{\gamma}_{A'A}):
  \mathcal{L}_{\M,\N} \mapsto \mathbb{C}^{\M\N \times \M\N}$, where
  $\mathcal{I}_{A'}$ is the identity map on $A'$ and where
  $\mathcal{L}_{\M,\N}$ denotes the set of all linear maps from
  $ \mathbb{C}^{\M \times \M}$ to $ \mathbb{C}^{\N \times \N}$. The
  Choi map sends a channel to its Choi-state.
  \end{defi}

\begin{defi} [Definition of the map $T$] \label{defi:map_T}
Let $V_{A\mapsto CB}\in \mathbb{C}^{\K\N \times \M}$, where the
systems $A$, $B$ and $C$ have the (complex) dimensions $\M$, $\N$ and
$\K$ respectively. We define the linear map $\mathcal{E}_{V_{A\mapsto
    CB}}(M_A)=\tr_C \, V_{A\mapsto CB} M_A V^{\conjTrans}_{A\mapsto CB}:
\mathbb{C}^{\M \times \M} \mapsto \mathbb{C}^{\N \times \N}$. This
allows us to define the smooth map $T(V_{A\mapsto CB})= \Gamma
\left(\mathcal{E}_{V_{A\mapsto CB}}\right):\mathbb{C}^{\K\N \times
  \M}\mapsto \mathbb{C}^{\M\N \times \M\N}$. 
  \end{defi}

  Note that the map $T$ sends a Stinespring dilation $V \in
  V_{\M,\K\N}$ of a channel $\mathcal{E} \in
  \mathcal{E}_{\M,\N,\leqslant \K}$ to the Choi-state representation
  of $\mathcal{E}$.

\begin{lem}  \label{lem:homeomorphism}
The manifolds $V_{\M,\N,\K}/U(\K)$ and $\mathcal{C}_{\M,\N,\K}$ are homeomorphic.
\end{lem}

\begin{proof}
  Since the Kraus rank of a channel is equal to the rank of the
  corresponding Choi-state~\cite{choi}, we can consider the map $T$ as
  a map from $V_{\M,\N,\K}$ to $\mathcal{C}_{\M,\N,\K}$. This map (which we
  still denote by $T$) is smooth by Proposition~\ref{prop:8.22}
  and~\ref{prop:8.25}. Let $\pi:V_{\M,\N,\K}\mapsto V_{\M,\N,\K}/U(\K)$
  denote the quotient map introduced in Lemma~\ref{lem:Quotient}. For
  all $U\in U(\K)$ we have
  $T((U\otimes I)V_{A\mapsto CB})=T(V_{A\mapsto CB})$,
  because the unitary action corresponds to a change of the basis of
  the system $C$, which is traced out.\footnote{We can also think of
    the unitary action as exploiting the unitary freedom on the Kraus
    representation, so the channel itself is unchanged under the
    unitary action.} In other words, the map $T$ is constant on the
  fibers of the quotient map $\pi$. By Theorem~4.30 of~\cite{Lee},
  there is a unique smooth map
  $\phi: V_{\M,\N,\K}/U(\K) \mapsto \mathcal{C}_{\M,\N,\K}$, such that
  the following diagram commutes.

\begin{tikzpicture}[description/.style={fill=white,inner sep=2pt}]
\matrix (m) [matrix of math nodes, row sep=3em,
column sep=4em, text height=1.5ex, text depth=0.25ex]
{ V_{\M,\N,\K} &  \\
 {V_{\M,\N,\K}/U(\K)} & \mathcal{C}_{\M,\N,\K}\\ };
%\draw[double,double distance=5pt] (m-1-1) Ð (m-1-3);
\draw[->,font=\scriptsize]
(m-1-1) edge node[auto] {$ T $} (m-2-2)
(m-1-1) edge node[auto] {$ \pi$} (m-2-1)
(m-2-1) edge node[auto] {$ \phi $} (m-2-2);
\end{tikzpicture}

By Lemma~\ref{lem:channels_isom} we have a one-to-one correspondence between $\mathcal{E}_{\M,\N,\K}$ and $ V_{\M,\N,\K}/U(\K)$ (which we denote by
$\mathcal{E}_{\M,\N,\K} \leftrightarrow V_{\M,\N,\K}/U(\K)$) and by the
Choi-Jamiolkowski isomorphism we have $\mathcal{E}_{\M,\N,\K} \leftrightarrow
\mathcal{C}_{\M,\N,\K}$. Together, this implies that $\phi$ is a
bijection.

We have left to show that $\phi^{-1}$ is continuous. We would like to
use the fact that a bijective continuous map from a compact space to a
Hausdorff space has a continuous inverse (cf.\ Lemma~A.52
of~\cite{Lee}).  To make our domain $V_{\M,\N,\K}/U(\K)$ compact, we
enlarge it to $V_{\M,\K\N}/U(\K)$. Let $\tilde{\pi}:V_{\M,\K\N}\mapsto
V_{\M,\K\N}/U(\K)$ denote the quotient map.\footnote{The action of the
  Lie group $U(\K)$ on $V_{\M,\K\N}$ is not free in general.} Since
$\tilde{\pi}$ is continuous and the Stiefel manifold $V_{\M,\K\N}$ is
compact, $\tilde{\pi}(V_{\M,\K\N})=V_{\M,\K\N}/U(\K)$ is compact. We
enlarge the domain of the map $T$ and denote this map by $\tilde{T}:
V_{\M,\K\N} \mapsto \mathcal{C}_{\M,\N,\leqslant \K}$. Note that
$\tilde{T}$ is continuous. Since $\tilde{T}$ is constant on the fibers
of $\tilde{\pi}$, we can define a map $\psi:V_{\M,\K\N}/U(\K)\mapsto
\mathcal{C}_{\M,\N,\leqslant \K}$, such that the following diagram
commutes.

\begin{tikzpicture}[description/.style={fill=white,inner sep=2pt}]
\matrix (m) [matrix of math nodes, row sep=3em,
column sep=4em, text height=1.5ex, text depth=0.25ex]
{ V_{\M,\K\N} &  \\
 {V_{\M,\K\N}/U(\K)} & \mathcal{C}_{\M,\N,\leqslant \K}\\ };
%\draw[double,double distance=5pt] (m-1-1) Ð (m-1-3);
\draw[->,font=\scriptsize]
(m-1-1) edge node[auto] {$ \tilde{T} $} (m-2-2)
(m-1-1) edge node[auto] {$ \tilde{\pi}$} (m-2-1)
(m-2-1) edge node[auto] {$ \psi $} (m-2-2);
\end{tikzpicture}

Note that $\psi$ is a bijection, because $\mathcal{E}_{\M,\N,\leqslant \K} \leftrightarrow V_{\M,\K\N}/U(\K)$ by Remark 4 of~\cite{choi}, and
$\mathcal{E}_{\M,\N,\leqslant \K} \leftrightarrow \mathcal{C}_{\M,\N,\leqslant \K}$ by
the Choi-Jamiolkowski isomorphism. The map $\psi$ is continuous,
because $\tilde{T}$ is continuous (and by the definition of the
quotient topology). Since $V_{\M,\K\N}/U(\K)$ is compact and
$\mathcal{C}_{\M,\N,\leqslant \K}$ is Hausdorff, $\psi^{-1}$ is
continuous.

To see that $\phi^{-1}$ is continuous we restrict $\psi^{-1}$ to
$\mathcal{C}_{\M,\N,\K}$.  For this purposes, we define the inclusion
map $\imath(V)=V:V_{\M,\N,\K}\mapsto V_{\M,\K\N}$. Note that $\imath$
is continuous and open (because $V_{\M,\N,\K}$ is an open subset of
$V_{\M,\K\N}$ by Proposition~\ref{prop:V_M,N,K_is_manifold}). Since
the map $ \tilde{\pi} \circ \imath $ is constant on the fibers of
$\pi$, we can define a map $ \tilde{\imath}$ such that the following
diagram commutes.

\begin{tikzpicture}[description/.style={fill=white,inner sep=2pt}]
\matrix (m) [matrix of math nodes, row sep=3em,
column sep=4em, text height=1.5ex, text depth=0.25ex]
{ V_{\M,\N,\K} & V_{\M,\K\N} \\
 {V_{\M,\N,\K}/U(\K)} & V_{\M,\K\N}/U(\K)\\ };
%\draw[double,double distance=5pt] (m-1-1) Ð (m-1-3);
\draw[->,font=\scriptsize]
(m-1-1) edge node[auto] {$ \imath $} (m-1-2)
(m-1-1) edge node[auto] {$ \tilde{\pi} \circ \imath $} (m-2-2)
(m-1-1) edge node[auto] {$ \pi$} (m-2-1)
(m-2-1) edge node[auto] {$ \tilde{\imath}$} (m-2-2)
(m-1-2) edge node[auto] {$\tilde{\pi}$} (m-2-2);
\end{tikzpicture}

By Lemma~\ref{lem:9.15}, $\tilde{\pi}$ is an open map. Then $
\tilde{\pi} \circ \imath $ is an open and continuous map and hence, we
can conclude that $ \tilde{\imath}$ is continuous and open (and
injective). We are now ready to show that $\phi^{-1}$ is
continuous. Note that the restriction $\tilde{\psi}^{-1}:
\mathcal{C}_{\M,\N,\K} \mapsto V_{\M,\K\N}/U(\K)$ of $\psi^{-1}$ is
still continuous. Because $ \tilde{\imath}$ is injective and
$\psi^{-1}(\mathcal{C}_{\M,\N,\K})=\tilde{\imath}(V_{\M,\N,\K}/U(\K))$
we can define a map $\chi:\mathcal{C}_{\M,\N,\K} \mapsto
V_{\M,\N,\K}/U(\K)$ such that the following diagram commutes.

\begin{tikzpicture}[description/.style={fill=white,inner sep=2pt}]
\matrix (m) [matrix of math nodes, row sep=3em,
column sep=4em, text height=1.5ex, text depth=0.25ex]
{ \mathcal{C}_{\M,\N,\K} & V_{\M,\K\N}/U(\K)  \\
 {} &V_{\M,\N,\K}/U(\K)\\ };
%\draw[double,double distance=5pt] (m-1-1) Ð (m-1-3);
\draw[->,font=\scriptsize]
(m-1-1) edge node[auto] {$ \tilde{\psi}^{-1} $} (m-1-2)
(m-1-1) edge node[auto] {$ \chi$} (m-2-2)
(m-2-2) edge node[auto] {$  \tilde{\imath} $} (m-1-2);
\end{tikzpicture}

The map $\chi$ is continuous because $\tilde{\imath}$ is open and note that $\chi=\phi^{-1}$. \end{proof}

\begin{thm} [Manifold structure for
  $\mathcal{C}^{\textnormal{e}}_{\M,\N,\K}$] \label{thm:C^e_{M,N,K}}
  The set $\mathcal{C}^{\textnormal{e}}_{\M,\N,\K}$ is an open subset
  of $\mathcal{C}_{\M,\N,\K}$. In particular it is a smooth embedded
  submanifold of $\mathcal{C}_{\M,\N,\K}$ (and of
  $\mathbb{R}^{2\M^2\N^2}$). Its dimension is
  $2\M\K\N-\K^2-\M^2$. Moreover,
  $\mathcal{C}^{\textnormal{e}}_{\M,\N,\M}\neq\emptyset$.
\end{thm}

\begin{proof}
  Follows from Lemma~\ref{lem:Extreme_channels} together with
  Lemma~\ref{lem:homeomorphism}.
\end{proof}

\begin{rmk} 
An alternative and more explicit characterization of extremality in the Choi-state representation is given in Theorem~4 in~\cite{Friedland}.
\end{rmk}

\section{Decomposition of channels into Convex Combinations of Extreme Channels} \label{sec:convex_comb}

We show that every element $\mathcal{E} \in \mathcal{E}_{\M,\N}$ can
be decomposed into a convex combination of at most $\M^2(\N^2-1)+1$
extreme channels in $\mathcal{E}^{\textnormal{e}}_{\M,\N}$.

\begin{thm}[Convex decomposition] \label{thm:convex} For every channel
  $\mathcal{E} \in \mathcal{E}_{\M,\N}$ there exists a set
  $\{(p_j,\mathcal{E}_j)\}_{j\in \{1,2,\dots,k\}}$, where $k\leqslant
  \M^2(\N^2-1)+1$, $p_j \in [0,1]$, $\sum_{j=1}^k p_j=1$ and
  $\mathcal{E}_j \in \mathcal{E}^{\textnormal{e}}_{\M,\N}$, such that
  $\mathcal{E}=\sum_{j=1}^k p_j \mathcal{E}_j$.
\end{thm}

\begin{rmk} 
  It is conjectured by Ruskai and Audenaert~\cite{Ruskai} that
  $k\leqslant \N$ if we allow convex combinations of channels
  $\mathcal{E}_j\in\mathcal{E}_{\M,\N,\leqslant \M}$ (note that
  $\mathcal{E}_{\M,\N,\leqslant \M}$ is equal to the closure of the
  set of all $\M$ to $\N$ extreme channels~\cite{Ruskai}). However, as
  far as we know, this remains unproven.
\end{rmk} 

\begin{proof}[Proof of Theorem~\ref{thm:convex}]
  In the proof of Lemma~\ref{lem:homeomorphism} we saw that
  $V_{\M,\M\N^2}/U(\M\N)$ is compact and homeomorphic to
  $\mathcal{C}_{\M,\N,\leqslant \M\N}$. Therefore,
  $\mathcal{C}_{\M,\N}=\mathcal{C}_{\M,\N,\leqslant \M\N}$ is
  compact. Since $\mathcal{C}_{\M,\N} \subset
  \mathbb{R}^{2\M^2{\N}^2}$ is also convex, by the Minkowski Theorem
  (see for example Theorem~2.3.4 of~\cite{Urruty}),
  $\mathcal{C}_{\M,\N}$ is the convex hull of its extreme points. By
  Carath\'eodory's theorem (see for example Theorem~1.3.6
  of~\cite{Urruty}), we can always find a decomposition
  of the required form for which
  $k\leqslant
  \textnormal{dim}\left(\textnormal{aff}[\mathcal{C}_{\M,\N}]\right)+1=\M^2(\N^2-1)+1$,
  where aff$[\mathcal{C}_{\M,\N}]$ denotes the affine hull of the set
  $\mathcal{C}_{\M,\N}$, i.e.,
  aff$[\mathcal{C}_{\M,\N}]=\{C_{AB}\in H_{\M\N}:\tr_B\,C_{AB}=\frac{1}{\M}I\}$.
\end{proof}

 \section{Application: Implementation of Quantum Channels} \label{sec:lower_bound}
 
 Methods for implementing quantum channels from a system $A$ to a
 system $B$ with low experimental cost as a sequence of
 simple-to-perform operations were considered in~\cite{Wang_qubit,
   Channels, Wang_qudit, Wang_new}. In~\cite{Channels}, a lower bound
 on the number of parameters required for a quantum circuit topology
 that is able to perform arbitrary extreme channels from $m$ to $n$
 qubits (i.e., from a system $A$ of dimension $\MM=2^m$ to a system
 $B$ of dimension $\NN=2^n$) was given. Here, we give a rigorous
 mathematical proof of this statement and strengthen the result by
 showing that a circuit topology that has
 fewer parameters than required by the lower bound is not able to approximate every extreme channel from $m$ to $n$ qubits arbitrarily well. \\

From a mathematical point of view, a quantum circuit topology can be
defined as follows.

\begin{defi}
  A \emph{quantum circuit topology} is a 5-tuple
  $Z:=(\MM,\NN,\CC,p,h)$, where $\MM,\NN,\CC \in \mathbb{N}$, $\NN\CC
  \geqslant \MM$, $p \in \mathbb{N}_0$ and $h:[0,2\pi]^p \mapsto
  V_{\MM,\NN\CC}$ is a smooth function.
\end{defi}

The physical interpretation is the following: We consider a quantum
system $BC$ of dimension $\LL:=\NN\CC$, where an input state for a
quantum channel is given on a subsystem $A$ of dimension $\MM$ and
where the other part of the system $BC$ starts in a fixed pure
state. We think of a fixed sequence of unitary operations performed on
the system $BC$, where the unitaries have $p$ free parameters between
them. Since the parameters corresponds to rotational angles
in~\cite{Channels}, we take them to lie in the interval
$[0,2\pi]$.\footnote{We could replace $2\pi$ by any positive real
  number without changing one of the following statements.} Each
choice of parameters corresponds to the implementation of a certain
isometry from the system $A$ to the system $BC$. The function $h$ maps
each choice of the parameters to the corresponding isometry. After
performing the isometry, the system $C$ is discarded (traced out), and
we read out the output of the channel on the remaining system $B$.
\begin{defi} \label{defi:generated_channels}
The set of quantum channels (in the Choi-state representation) that
can be generated by the quantum circuit topology $Z=(\MM,\NN,\CC,p,h)$
is defined by $H(Z):=T(h([0,2\pi]^p))$, where the map $T$ was
introduced in Definition~\ref{defi:map_T} and where we take the
partial trace over the first $d_C$-dimensional system, i.e., the
partial trace $\tr_C(\cdot)$ corresponds to $\sum_{i=1}^{d_C}
\left(\bra{i} \otimes I\right) \cdot  \left(\ket{i} \otimes
  I\right)$.\footnote{Note that this specification does not restrict
  the physical setting since we can always adapt the map $h$, such that the output of a channel is read out at the last $d_B$ dimensional system.}
\end{defi}

\begin{lem} \label{lem:measure_zero} Let $\K \in \mathbb{N}$ be fixed
  and $O \subset \mathcal{C}_{\MM,\NN,\K}$ open (and non empty). A
  quantum circuit topology $Z=(\MM,\NN,\CC,p,h)$ with $p<
  \textnormal{dim}\left(\mathcal{C}_{\MM,\NN,\K}\right)=2\MM\NN\K-\MM^2-\K^2$
  or $\CC<\K$ can only generate a set of measure zero in $O$, i.e.,
  $H(Z) \cap O$ is of measure zero in $O$.
\end{lem}

\begin{proof}
  The idea of the proof is based on Sard's theorem (similar
  to~\cite{unitary_lowerb1,unitary_lowerb2}). Let us fix a quantum
  circuit topology $Z=(\MM,\NN,\CC,p,h)$, $\K \in \mathbb{N}$ and an
  open set $O \subset \mathcal{C}_{\MM,\NN,\K}$. We define the map
  $T:V_{\MM,\LL}\mapsto \mathcal{C}_{\MM,\NN}$ as in
  Definition~\ref{defi:map_T}, and a map $F=T\circ h$, such that the
  following diagram commutes.

\begin{tikzpicture}[description/.style={fill=white,inner sep=2pt}]
\matrix (m) [matrix of math nodes, row sep=3em,
column sep=4em, text height=1.5ex, text depth=0.25ex]
{ V_{\MM,\LL} & \\
 {[0,2\pi]^p} &  \mathcal{C}_{\MM,\NN}\\ };
%\draw[double,double distance=5pt] (m-1-1) Ð (m-1-3);
\draw[->,font=\scriptsize]
(m-1-1) edge node[auto] {$T $} (m-2-2)
(m-2-1) edge node[auto] {$ h$} (m-1-1)
(m-2-1) edge node[auto] {$ F $} (m-2-2);
\end{tikzpicture}

\textbf{{Case 1 ($\CC<\K$):}}
In this case, note that $H(Z) =F([0,2\pi]^p) \subset T(V_{\MM,\LL})$. But $T(V_{\MM,\LL})$ contains only (Choi-states of) channels of Kraus rank at most $\CC<\K$. Therefore, $H(Z) \cap O= \emptyset$.\\

\textbf{{Case 2 ($\CC\geqslant \K$}):} 
Define the set $S:=F([0,2\pi]^p)\cap O$. To show that $S$ has measure
zero, define the domain $D=F^{-1}(O)$ and the function $\tilde{F}=
\restr{F}{D}: D \mapsto O$. By Sard's theorem (see
Appendix~\ref{app:sard} for the full technical details) we conclude
that $S=\tilde{F}(D)$ is of measure zero in the smooth submanifold $O$
if $\textnormal{dim}(D)\leqslant
p<\textnormal{dim}(O)=\textnormal{dim}(\mathcal{C}_{\M,\N,\K})$.
\end{proof}

\begin{thm} [Strong lower bound] \label{cor:approx} Let $q \in
  \mathbb{N}$ and consider a set of quantum circuit topologies
  $R=\{Z_i=(\MM,\NN,\CCi,p_i,h_i) \}_{i \in \{1,2,\dots, q\}}$ where
  for each $i\in\{1,2,\dots,q\}$ either $p_i<2\MM^2\left(\NN-1\right)$
  or $\CCi<\MM$. Then there exists an extreme channel $ \mathcal{E}_0
  \in \mathcal{C}^{\textnormal{e}}_{\MM,\NN,\MM}$ and a neighborhood
  $B( \mathcal{E}_0) \subset \mathcal{C}_{\MM,\NN}$ of
  $\mathcal{E}_0$, such that for all $\mathcal{E} \in B(
  \mathcal{E}_0)$ we have $\mathcal{E} \notin \bigcup_{i=1}^q H(Z_i)$.
\end{thm}

This theorem considers a finite set of circuit topologies each of
which either has fewer free parameters than the dimension of the set
of extreme channels or discards a system whose dimension is too low to
generate channels of the maximal Kraus rank for any extreme channel.
The theorem says that there exist extreme channels that cannot be
approximated arbitrarily well using circuit topologies from this set.

\begin{proof}
  By Theorem~\ref{thm:C^e_{M,N,K}}, the set
  $\mathcal{C}^{\textnormal{e}}_{\MM,\NN,\MM}\neq \emptyset$ is an
  open subset of $\mathcal{C}_{\MM,\NN,\MM}$. Hence, $H(Z_i) \cap
  \mathcal{C}^{\textnormal{e}}_{\MM,\NN,\MM}$ is of measure zero in
  $\mathcal{C}^{\textnormal{e}}_{\MM,\NN,\MM}$ by
  Lemma~\ref{lem:measure_zero}. Since a finite union of set of measure
  zero is again of measure zero, we conclude that $S:=\bigcup_{i=1}^q
  \left(H(Z_i) \cap \mathcal{C}^{\textnormal{e}}_{\MM,\NN,\MM}
  \right)$ is of measure zero. Since
  $\mathcal{C}^{\textnormal{e}}_{\MM,\NN,\MM}\neq \emptyset$, we can
  choose a channel $\mathcal{E}_0 \in
  \mathcal{C}^{\textnormal{e}}_{\MM,\NN,\MM} \cap S^{\textnormal{c}}$
  and hence  $\mathcal{E}_0 \in H^{\textnormal{c}}$, where we set
  $H:=\bigcup_{i=1}^q H(Z_i)$. We have left to show that $H$ is
  closed in $\mathcal{C}_{\MM,\NN}$. To see this, note that the map
  $T_i:V_{\MM,\LLi} \mapsto C_{\MM,\NN}$ (as defined in
  Definition~\ref{defi:map_T}) is continuous and hence $F_i:=T_i \circ
  h_i:[0,2 \pi]^{p_i} \mapsto C_{\MM,\NN}$ is also continuous. Since
  $[0,2 \pi]^{p_i}$ is compact, $H(Z_i)=F_i([0,2 \pi]^{p_i})$ is
  closed in $\mathcal{C}_{\MM,\NN}$ and therefore $H$ is also closed.
\end{proof}

\section{Acknowledgements}
We thank Matthias Christandl for useful discussions. R.C.\
acknowledges support from the EPSRC's Quantum Communications Hub (grant number EP/M013472/1).

\appendix

\section{Smoothness of $\tilde{F}$ and a version of Sard's theorem}\label{app:sard}

The domain $D$ of the function $\tilde{F}$ in the proof of
Lemma~\ref{lem:measure_zero} might not be open. We therefore first
clarify the meaning of smoothness in the case of arbitrary domains.

\begin{defi}\label{def:smoothness} 
  Let $D \subset \mathbb{R}^{\M}$. A function
  $F:D \mapsto \mathbb{R}^{\N}$ is \emph{smooth} if for all $p \in D$
  there exists a neighborhood $B(p)$ of $p$ in $ \mathbb{R}^{\M}$, such
  that there exists an extension
  $\hat{F}:B(p) \mapsto \mathbb{R}^{\N}$ of $F$, with $\hat{F}$
  smooth.
\end{defi}

\begin{lem}[Measure of the image] \label{lem:measure_zero_image} Let
  $\M<\N$, $D \subset \mathbb{R}^{\M}$ and $F:D \mapsto
  \mathbb{R}^{\N}$ be smooth. Then $F(D)$ has measure zero in
  $\mathbb{R}^{\N}$.
\end{lem}
\begin{proof} We define
  $D':=D\times\{0\}\times\{0\}\times\dots\times\{0\}\subset\mathbb{R}^{\N}$. Note
  that $D'$ lies in a affine subspace of $\mathbb{R}^{\N}$ and is
  therefore of measure zero (in $\mathbb{R}^{\N}$). Let $F':=F \circ
  \pi:D' \mapsto \mathbb{R}^{\N}$, where $\pi:\mathbb{R}^{\N} \mapsto
  \mathbb{R}^{\M}$ denotes the projection map to the first $\M$
  coordinates. To show that $F'$ is smooth, choose $p' \in D'$ and let
  $p:=\pi(p')$. Because $F$ is smooth by assumption, there exists a
  neighborhood $B(p)\subset \mathbb{R}^{\M}$ around the point $p$ and
  a smooth extension of $F$ denoted by $\hat{F}:B(p) \rightarrow
  \mathbb{R}^{\N}$. Hence $\hat{F}':=\hat{F}\circ\pi:B(p)\times
  \mathbb{R}^{\N-\M} \mapsto \mathbb{R}^{\N}$ is a smooth extension of
  $F'$ around $p'$. Therefore, $F'$ is a smooth map and by
  Proposition~6.5 of~\cite{Lee}, we conclude that $F'(D')=F(D)$ is of
  measure zero.
   \end{proof} 
  
\begin{lem}[Restrict the range of a smooth map]\label{lem:restrict_range}
  Let $D \subset \mathbb{R}^{\M}$ be arbitrary and let $N$ be a smooth
  manifold. Let $N' $ be a smooth embedded submanifold of $N$ and $F:D
  \mapsto N$ be a smooth map, such that $F(D) \subset N'$. Then
  $\tilde{F}:D \mapsto N'$ is smooth.
\end{lem}

\begin{proof}The proof works analogously to the proof of Theorem~5.29
  of~\cite{Lee} (see also Corollary~5.30 of~\cite{Lee}).
  \end{proof}

\begin{proof1}[Proof:  $\tilde{F}(D)$ is of measure zero (completes the proof of case 2 of Lemma~\ref{lem:measure_zero})]
  We use the notation of the proof of
  Lemma~\ref{lem:measure_zero}. First note that the function
  $\restr{F}{D}:D \mapsto \mathbb{C}^{\MM\NN \times \MM\NN}$ is
  smooth. By Lemma~\ref{lem:restrict_range} the function $\tilde{F}: D
  \mapsto O $ is also smooth. To show that $S:=\tilde{F}(D) \subset O
  $ is of measure zero, we choose a collection of smooth charts
  $\{(U_{\alpha},\phi_{\alpha})\}$ of the submanifold $O$ whose
  domains cover $S$. By Lemma~6.6 of~\cite{Lee} we have left to show
  that for all $\alpha$ the image $\phi_{\alpha}(S \cap U_{\alpha})$
  is of measure zero in $ \mathbb{R}^{d}$, where $d$ denotes the
  dimension of $O$. Consider the smooth map
  $\tilde{F}_{\alpha}:=\phi_{\alpha} \circ \tilde{F}: D_{\alpha}:=
  \tilde{F}^{-1}(U_{\alpha}) \rightarrow
  V_{\alpha}:=\phi_{\alpha}(U_{\alpha}) \subset \mathbb{R}^d$.
    
  By Lemma~\ref{lem:measure_zero_image}, the image
  $\tilde{F}_{\alpha}(D_{\alpha})=\phi_{\alpha}(S \cap U_{\alpha})$ is
  of measure zero if $p<d$.
\end{proof1}

\end{document}